\documentclass[journal,comsoc,twocolumn,letterpaper]{IEEEtran}

\usepackage[utf8]{inputenc}
\usepackage{amsmath,amssymb,amsfonts}
\usepackage{amsthm}
\usepackage{graphicx}
\usepackage{textcomp}
\usepackage{url}
\usepackage{cite}
\usepackage[hidelinks]{hyperref}   

\newtheorem{theorem}{Theorem}
\newtheorem{lemma}{Lemma}
\newtheorem{proposition}{Proposition}

\newtheorem{definition}{Definition}
\newtheorem{remark}{Remark}

\newcommand{\Fq}{\mathbb{F}_q}
\newcommand{\Fqq}{\mathbb{F}_{q^2}}
\newcommand{\wt}{\mathrm{wt}}
\newcommand{\kopt}{k_{\mathrm{opt}}^{(q)}}
\newcommand{\nopt}{n_{\mathrm{opt}}^{(q)}}
\newcommand{\dopt}{d_{\mathrm{opt}}^{(q)}}

\begin{document}

\title{Alphabet-Dependent Bounds for Pure Quantum\\ $(r,\rho)$-Locally Recoverable Codes}

\author{Vijay~Kumar~and~Ramakrishna~Bandi%
\thanks{Manuscript received XXXX XX, 20XX; revised XXXX XX, 20XX. The associate editor coordinating the review of this letter and approving it for publication was X.~X.~Xxxxx.}%
\thanks{The authors are with the Department of Science and Mathematics, International Institute of Information Technology, Naya Raipur 493661, India (e-mail: vijayk@iiitnr.edu.in; ramakrishna@iiitnr.edu.in).}%
\thanks{Digital Object Identifier 10.1109/LCOMM.20XX.XXXXXXX}}

\maketitle

\begin{abstract}
A quantum $(r,\rho)$-locally recoverable code ($(r,\rho)$-qLRC) is a quantum
code in which every qudit can be recovered from at most $r+\rho-1$ other
qudits, even after $\rho-1$ additional erasures inside the recovery set. The
bounds currently known for this class, namely the Singleton-like and the
GG Singleton-like bounds, are alphabet independent and are therefore loose for
small-to-moderate qudit dimensions. In this letter, we derive three
alphabet-dependent upper bounds for pure $(r,\rho)$-qLRCs obtained through the
Hermitian CSS construction: a Griesmer-like, a Plotkin-like, and a
sphere-packing-like bound. We further establish the asymptotic hierarchy among
these bounds and identify the relative-distance regions in which each of them
yields the tightest rate constraint.
\end{abstract}

\begin{IEEEkeywords}
Quantum locally recoverable codes, alphabet-dependent bounds, Griesmer bound,
Plotkin bound, sphere-packing bound, CSS codes, rate--distance trade-off.
\end{IEEEkeywords}

\IEEEpeerreviewmaketitle

\section{Introduction}
\IEEEPARstart{L}{ocally} recoverable codes (LRCs) allow an erased symbol to be
reconstructed from a small subset of the remaining symbols rather than from the
entire codeword, and are foundational to distributed and cloud storage
\cite{gopalan2012locality}. The $(r,\rho)$ generalization, in which each
recovery set tolerates $\rho-1$ erasures while using at most $r+\rho-1$
symbols, was introduced in \cite{prakash2012optimal}. For a $q$-ary
$(r,\rho)$-LRC with parameters $[n,k,d]$, the minimum distance satisfies the
Singleton-like bound
\begin{equation}\label{eq:cSingleton}
d \leq n-k+1-\left(\left\lceil\frac{k}{r}\right\rceil-1\right)(\rho-1),
\end{equation}
which depends only on $n$, $k$, $r$ and $\rho$, and hence does not capture
finite-alphabet effects. Grezet \textit{et al.} \cite{grezet2019alphabet}
tightened \eqref{eq:cSingleton} by deriving an alphabet-dependent bound of the
form
\begin{equation}\label{eq:CMbound}
k \leq \min_{\tau\in\mathbb{Z}_{+}}
\left\{\tau r+\kopt\bigl(n-\tau(r+\rho-1),\,d\bigr)\right\},
\end{equation}
where $\kopt(n,d)$ denotes the largest dimension of a $q$-ary linear code with
the remaining two parameters fixed. Gao \textit{et al.}
\cite{gao2025griesmer}, generalizing the LRC bounds of \cite{hao2020bounds},
combined Cadambe--Mazumdar-type \cite{cadambe2015bounds} reductions with the
classical Griesmer, sphere-packing and Plotkin bounds
\cite{huffman2010fundamentals} to obtain explicit alphabet-dependent bounds for
$(r,\rho)$-LRCs.

The quantum counterpart was first studied by Golowich and Guruswami
\cite{Golowic}, who introduced quantum LRCs (qLRCs) with locality $r$,
established a Singleton-like bound, and gave near-optimal constructions based on
folded quantum Tamo--Barg codes. Galindo \textit{et al.}
\cite{galindo2026quantum} extended this framework to $(r,\rho)$-qLRCs capable of
correcting $\rho-1$ qudit erasures per recovery set, characterized stabilizer
and CSS constructions, and established a Singleton-like bound together with
optimal pure examples. Several subsequent works have constructed optimal
$(r,\rho)$-qLRCs
\cite{cao2025optimal,zhou2025optimal,galindo2023optimal,guruswami2026quantum,cao2026quantum}.
In particular, \cite[Th.~31]{galindo2026quantum} establishes the
\emph{Singleton-like} bound for pure $(r,\rho)$-qLRCs,
\begin{equation}\label{eq:qSingleton}
2\delta \leq n-\kappa-2\left(\left\lceil\frac{n+\kappa}{2r}\right\rceil-1\right)
(\rho-1)+2,
\end{equation}
whereas \cite[Th.~35]{Golowic} establishes a general Singleton-like bound for
arbitrary qLRCs, which extends to $(r,\rho)$-qLRCs and is referred to here as
the \emph{GG Singleton-like} bound:
\begin{equation}\label{eq:GGSingleton}
\begin{aligned}
\kappa \leq {}& n-2(\delta-1)
-\left\lfloor\frac{n-(\delta-1)}{r+\rho-1}\right\rfloor\\[1mm]
&-\left\lfloor
\frac{n-2(\delta-1)-\left\lfloor\frac{n-(\delta-1)}{r+\rho-1}\right\rfloor}
{r+\rho-1}\right\rfloor .
\end{aligned}
\end{equation}
Like their classical counterparts, neither \eqref{eq:qSingleton} nor
\eqref{eq:GGSingleton} explicitly incorporates the qudit dimension $q$. In this
letter we address this gap by deriving quantum \emph{Griesmer-like},
\emph{sphere-packing-like} and \emph{Plotkin-like} bounds for pure
$(r,\rho)$-qLRCs via the Hermitian CSS construction, and by characterizing their
asymptotic behaviour and relative tightness with respect to the
\emph{Singleton-like} and \emph{GG Singleton-like} bounds.

The rest of the letter is organized as follows. Section~\ref{sec:prelim} recalls
the required notation, definitions and known results. Section~\ref{sec:adb}
presents the three alphabet-dependent bounds for $(r,\rho)$-qLRCs and analyzes
their asymptotic form. Section~\ref{sec:conclusion} concludes the letter.

\section{Preliminaries}\label{sec:prelim}

Throughout this letter, the locality parameters are written as $(r,\rho)$
rather than as $(r,\delta)$ as in \cite{galindo2026quantum}, so as to
distinguish the locality parameter from the minimum distance $\delta$ of a
quantum error-correcting code (QECC). For a subset $S\subseteq\Fq^{\,n}$ we
write $\wt(S)=\min_{\mathbf{u}\in S,\ \mathbf{u}\neq\mathbf{0}}\wt(\mathbf{u})$,
and for a linear code $\mathcal{C}$ we let $\wt(\mathcal{C})$ denote the minimum
weight of a nonzero codeword.

\subsection{Classical \texorpdfstring{$(r,\rho)$}{(r,rho)}-LRCs}

An $[n,k,d]_q$ linear code is a $k$-dimensional subspace of $\Fq^{\,n}$ with
block length $n$, dimension $k$ and minimum Hamming distance
$d=\wt(\mathcal{C})$. For a code $\mathcal{C}$ over $\Fq$, its Euclidean dual is
$\mathcal{C}^{\perp_E}=\{\mathbf{x}\in\Fq^{\,n}:\langle\mathbf{x},\mathbf{y}
\rangle=\sum_i x_iy_i=0,\ \forall\,\mathbf{y}\in\mathcal{C}\}$; for a code over
$\Fqq$, its Hermitian dual is
$\mathcal{C}^{\perp_H}=\{\mathbf{x}\in\Fqq^{\,n}:\langle\mathbf{x},\mathbf{y}
\rangle_H=\sum_i x_iy_i^{q}=0,\ \forall\,\mathbf{y}\in\mathcal{C}\}$. The code
$\mathcal{C}$ is Hermitian (resp.\ Euclidean) \emph{dual containing} if
$\mathcal{C}^{\perp_H}\subseteq\mathcal{C}$ (resp.\
$\mathcal{C}^{\perp_E}\subseteq\mathcal{C}$).

\begin{definition}
For $i\in\{1,\dots,n\}$, the coordinate $i$ of $\mathcal{C}$ has
\emph{$(r,\rho)$-locality} if it lies in a recovery set
$S_i\subseteq\{1,\dots,n\}$ with $|S_i|\leq r+\rho-1$ and
$d(\mathcal{C}|_{S_i})\geq\rho$, where $\mathcal{C}|_{S_i}$ denotes
$\mathcal{C}$ punctured to $S_i$. The code $\mathcal{C}$ is an
$(r,\rho)$-LRC if this holds for every $i$.
\end{definition}

Gao \textit{et al.} \cite{gao2025griesmer} established the following
alphabet-dependent bounds for $(r,\rho)$-LRCs by generalizing the LRC bounds of
\cite[Prop.~1]{hao2020bounds}.

\begin{lemma}[{\cite[Th.~2]{gao2025griesmer}}]\label{lem:cADB}
Let $1\leq r<k$ and $\rho\geq2$, and let $\mathcal{C}$ be a $q$-ary
$(r,\rho)$-LRC with parameters $[n,k,d]$. Put
$\tau_{\max}=\lceil k/r\rceil-1$. Then
\begin{align}
n &\geq \max_{0\leq\tau\leq\tau_{\max}}
\left\{\tau(r+\rho-1)+\nopt(k-\tau r,\,d)\right\},\label{eq:ADgriesmer}\\
k &\leq \min_{0\leq\tau\leq\tau_{\max}}
\left\{\tau r+\kopt\bigl(n-\tau(r+\rho-1),\,d\bigr)\right\},\label{eq:ADsphere}\\
d &\leq \min_{0\leq\tau\leq\tau_{\max}}
\left\{\dopt\bigl(n-\tau(r+\rho-1),\,k-\tau r\bigr)\right\},\label{eq:ADplotkin}
\end{align}
where $\nopt(k,d)$, $\kopt(n,d)$ and $\dopt(n,k)$ denote, respectively, the
smallest length, the largest dimension and the largest minimum distance of a
$q$-ary linear code with the two remaining parameters fixed.
\end{lemma}

As a consequence of Lemma~\ref{lem:cADB}, we obtain the following explicit
bounds for $(r,\rho)$-LRCs. For brevity we write
\begin{equation}\label{eq:tausp}
\tau_{\max}^{\mathrm{sp}}
=\min\left\{\tau_{\max},\left\lfloor\frac{n-1}{r+\rho-1}\right\rfloor\right\},
\end{equation}
which is the largest shortening index for which the residual length
$n-\tau(r+\rho-1)$ remains positive.

\begin{theorem}\label{th:cLRCbound}
Let $\mathcal{C}$ be a $q$-ary $(r,\rho)$-LRC with parameters $[n,k,d]$. Then
the following bounds hold.
\begin{enumerate}
\item Griesmer-like bound:
\begin{equation}\label{eq:cGriesmer}
n\geq\max_{0\leq\tau\leq\tau_{\max}}
\left\{\tau(r+\rho-1)+\sum_{i=0}^{k-\tau r-1}
\left\lceil\frac{d}{q^{i}}\right\rceil\right\}.
\end{equation}
\item Sphere-packing-like bound:
\begin{equation}\label{eq:cSphere}
\begin{aligned}
k\leq{}& n-\max_{0\leq\tau\leq\tau_{\max}^{\mathrm{sp}}}\Bigg\{\tau(\rho-1)\\
&+\log_{q}\left(\sum_{i=0}^{\left\lfloor\frac{d-1}{2}\right\rfloor}
\binom{n-\tau(r+\rho-1)}{i}(q-1)^{i}\right)\Bigg\}.
\end{aligned}
\end{equation}
\item Plotkin-like bound:
\begin{equation}\label{eq:cPlotkin}
d\leq\min_{0\leq\tau\leq\tau_{\max}}
\left\{\frac{q^{\,k-\tau r-1}(q-1)\bigl(n-\tau(r+\rho-1)\bigr)}
{q^{\,k-\tau r}-1}\right\}.
\end{equation}
\end{enumerate}
\end{theorem}

\begin{proof}
Assertion~1) is \cite[Cor.~3]{gao2025griesmer}. For~2), apply the classical
sphere-packing bound \cite[Th.~1.12.1]{huffman2010fundamentals} to a code of
length $n'=n-\tau(r+\rho-1)$ and dimension $k'=k-\tau r$ in
\eqref{eq:ADsphere}, which gives
$k-\tau r\leq n'-\log_q V$ with $V$ the Hamming-ball volume appearing in
\eqref{eq:cSphere}; rearranging and using
$n'=n-\tau(r+\rho-1)$ yields \eqref{eq:cSphere}. The requirement $n'\geq1$
forces $\tau\leq\lfloor(n-1)/(r+\rho-1)\rfloor$, whence the range
\eqref{eq:tausp}. Assertion~3) follows by applying the classical Plotkin bound
\cite[Th.~2.2.1]{huffman2010fundamentals} to the effective parameters $n'$ and
$k'$ in \eqref{eq:ADplotkin}.
\end{proof}

\subsection{Quantum \texorpdfstring{$(r,\rho)$}{(r,rho)}-LRCs}

A $q$-ary QECC $[[n,\kappa,\delta]]_q$ is a $q^{\kappa}$-dimensional subspace of
$(\mathbb{C}^{q})^{\otimes n}$ that corrects arbitrary errors on up to
$\lfloor(\delta-1)/2\rfloor$ qudits; its parameters obey the quantum Singleton
bound $2\delta\leq n-\kappa+2$, with equality defining a quantum maximum
distance separable code.

\begin{proposition}[Hermitian CSS construction \cite{calderbank1996good,steane1996error}]
\label{prop:css}
For $i=1,2$, let $\mathcal{C}_i$ be an $[n,k_i,d_i]_{q^2}$ linear code such that
$\mathcal{C}_2^{\perp_H}\subseteq\mathcal{C}_1$. Then there exists an
$[[n,k_1+k_2-n,\delta]]_q$ quantum code with
\[
\delta=\begin{cases}
d_1, & \text{if } \mathcal{C}_2^{\perp_H}=\mathcal{C}_1,\\[1mm]
\min\bigl\{\wt(\mathcal{C}_1\setminus\mathcal{C}_2^{\perp_H}),\,
\wt(\mathcal{C}_2\setminus\mathcal{C}_1^{\perp_H})\bigr\},
& \text{if } \mathcal{C}_2^{\perp_H}\subsetneq\mathcal{C}_1.
\end{cases}
\]
\end{proposition}

We next recall the definition of an $(r,\rho)$-qLRC, following the framework of
\cite{galindo2026quantum}.

\begin{definition}[{\cite[Defs.~9 and~10]{galindo2026quantum}}]\label{def:qlrc}
A QECC $\mathcal{Q}\subseteq(\mathbb{C}^{q})^{\otimes n}$ is an
\emph{$(r,\rho)$-qLRC} if, for each $i\in\{1,\dots,n\}$, there exists a set
$J\subseteq\{1,\dots,n\}$ containing $i$ with $|J|\leq r+\rho-1$ such that, for
every subset $I\subseteq J$ with $|I|=\rho-1$, there exists a trace-preserving
map $\mathcal{R}_{Q,I}^{J}$ acting only on the qudits indexed by $J$ and leaving
the remaining ones untouched, for which
\[
\mathcal{R}_{Q,I}^{J}\circ\Gamma^{I}
\bigl(|\psi\rangle\langle\psi|\bigr)=|\psi\rangle\langle\psi|
\]
for every $|\psi\rangle\in\mathcal{Q}$, where $\Gamma^{I}$ is defined as in
\cite[eq.~(5)]{galindo2026quantum}.
\end{definition}

The following result of Galindo \textit{et al.} \cite{galindo2026quantum}
characterizes the $(r,\rho)$-qLRCs obtained from Hermitian dual-containing
classical linear codes.

\begin{theorem}[{\cite[Th.~29]{galindo2026quantum}}]\label{th:qLRCs}
If $\mathcal{C}$ is an $[n,k,d]_{q^2}$ Hermitian dual-containing code with
$(r,\rho)$-locality and $d^{\perp_H}\geq\rho$, then there exists an
$[[n,\kappa,\delta]]_q$ $(r,\rho)$-qLRC with the same $(r,\rho)$-locality, where
$\kappa=2k-n$ and
\[
\delta=\begin{cases}
d, & \text{if } \mathcal{C}^{\perp_H}=\mathcal{C},\\[1mm]
\wt(\mathcal{C}\setminus\mathcal{C}^{\perp_H})\geq d,
& \text{if } \mathcal{C}^{\perp_H}\subsetneq\mathcal{C}.
\end{cases}
\]
Such an $(r,\rho)$-qLRC is said to be \emph{pure} if $\delta=d$, and
\emph{impure} otherwise.
\end{theorem}

\section{Alphabet-Dependent Bounds for Pure
\texorpdfstring{$(r,\rho)$}{(r,rho)}-qLRCs}\label{sec:adb}

In this section we derive three alphabet-dependent bounds for pure
$(r,\rho)$-qLRCs and analyze their asymptotic behaviour. Throughout, we write
\begin{equation}\label{eq:qtau}
T=\left\lceil\frac{n+\kappa}{2r}\right\rceil-1,\quad
T^{\mathrm{sp}}=\min\left\{T,\left\lfloor\frac{n-1}{r+\rho-1}\right\rfloor\right\}.
\end{equation}

\begin{remark}
If an $[[n,\kappa,\delta]]_q$ code arises from Theorem~\ref{th:qLRCs}, then
$\kappa=2k-n$, so $n+\kappa=2k$ is even and $(n+\kappa)/2=k$ is a positive
integer. All the expressions below are therefore well defined.
\end{remark}

\begin{theorem}\label{th:qADB}
Let $\mathcal{Q}$ be a pure $[[n,\kappa,\delta]]_q$ qLRC with
$(r,\rho)$-locality obtained from the Hermitian construction of
Theorem~\ref{th:qLRCs}. Then the following bounds hold.
\begin{enumerate}
\item Pure Griesmer-like bound:
\begin{equation}\label{eq:qGriesmer}
n\geq\max_{0\leq\tau\leq T}
\left\{\tau(r+\rho-1)+\sum_{i=0}^{\frac{n+\kappa}{2}-\tau r-1}
\left\lceil\frac{\delta}{q^{2i}}\right\rceil\right\}.
\end{equation}
\item Pure sphere-packing-like bound:
\begin{equation}\label{eq:qSphere}
\begin{aligned}
\kappa\leq{}& n-2\max_{0\leq\tau\leq T^{\mathrm{sp}}}\Bigg\{\tau(\rho-1)\\
&+\log_{q^2}\left(\sum_{i=0}^{\left\lfloor\frac{\delta-1}{2}\right\rfloor}
\binom{n-\tau(r+\rho-1)}{i}(q^{2}-1)^{i}\right)\Bigg\}.
\end{aligned}
\end{equation}
\item Pure Plotkin-like bound:
\begin{IEEEeqnarray}{rCl}
\delta &\leq& \min_{0\leq\tau\leq T}\Bigg\{\IEEEnonumber\\*
&&\qquad\frac{q^{\,n+\kappa-2\tau r-2}(q^{2}-1)
\bigl(n-\tau(r+\rho-1)\bigr)}
{q^{\,n+\kappa-2\tau r}-1}\Bigg\}.\IEEEeqnarraynumspace\label{eq:qPlotkin}
\end{IEEEeqnarray}
\end{enumerate}
\end{theorem}

\begin{proof}
By Theorem~\ref{th:qLRCs}, the existence of a pure $[[n,\kappa,\delta]]_q$ qLRC
$\mathcal{Q}$ with $(r,\rho)$-locality implies the existence of a classical
Hermitian dual-containing linear code $\mathcal{C}$ over $\Fqq$ with
$(r,\rho)$-locality, parameters
\begin{equation}\label{eq:qparam}
\left[n,\ \frac{n+\kappa}{2},\ \delta\right]_{q^{2}},
\end{equation}
and $d^{\perp_H}\geq\rho$. Since the locality of a linear code cannot exceed its
dimension, $r\leq(n+\kappa)/2$ and hence
$\lceil(n+\kappa)/(2r)\rceil\geq1$, so the ranges in \eqref{eq:qtau} are
nonempty. Applying the Griesmer-like bound \eqref{eq:cGriesmer} to
$\mathcal{C}$ with the parameters \eqref{eq:qparam}, that is, with $q$ replaced
by $q^{2}$, $k$ by $(n+\kappa)/2$ and $d$ by $\delta$, gives
\eqref{eq:qGriesmer}. Equivalently, on substituting $t=2i$, the inner sum in
\eqref{eq:qGriesmer} runs over the even integers
$0\leq t\leq n+\kappa-2\tau r-2$. Applying the sphere-packing-like bound
\eqref{eq:cSphere} and the Plotkin-like bound \eqref{eq:cPlotkin} to the same
code yields \eqref{eq:qSphere} and \eqref{eq:qPlotkin}, respectively.
\end{proof}

\subsection{Asymptotic Analysis}

The asymptotic regime, in which the code length tends to infinity, provides a
natural framework for comparing the tightness of different bounds. Let
\[
\mathcal{R}=\lim_{n\to\infty}\frac{\kappa}{n},\qquad
\Delta=\lim_{n\to\infty}\frac{\delta}{n}
\]
denote the rate and the relative distance, respectively, and set
$M=r+\rho-1$. With this notation, the \emph{GG Singleton-like} bound
\eqref{eq:GGSingleton} admits the asymptotic form established in
\cite[eq.~(11)]{luo2025bounds},
\begin{equation}\label{eq:aGG}
\mathcal{R}\leq\left(\frac{M-1}{M}\right)^{2}
-\frac{(2M-1)(M-1)}{M^{2}}\,\Delta+o(1).
\end{equation}

We next derive the asymptotic form of the pure \emph{Singleton-like} bound
\eqref{eq:qSingleton}. Using $\lceil x\rceil\geq x$ in \eqref{eq:qSingleton},
\[
2\delta\leq n-\kappa+2-2(\rho-1)\left(\frac{n+\kappa}{2r}\right)+2(\rho-1),
\]
and dividing by $n$ and letting $n\to\infty$ gives
\begin{equation}\label{eq:aSingleton}
\mathcal{R}\leq\frac{r-\rho+1}{M}-\frac{2r}{M}\,\Delta+o(1).
\end{equation}

Extending the asymptotic analysis of \cite{li2025optimal} and
\cite{li2025improved}, we now obtain the asymptotic forms of the
alphabet-dependent bounds of Theorem~\ref{th:qADB}. For the pure
\emph{Griesmer-like} bound \eqref{eq:qGriesmer}, evaluating at the largest
admissible index $\tau=\frac{n+\kappa}{2r}-1$ gives
\[
n\geq\left(\frac{n+\kappa}{2r}-1\right)M+\sum_{i=0}^{r-1}
\left\lceil\frac{\delta}{q^{2i}}\right\rceil .
\]
For any fixed integer $t\in[1,r]$,
\begin{IEEEeqnarray*}{rCl}
\sum_{i=0}^{r-1}\left\lceil\frac{\delta}{q^{2i}}\right\rceil
&\geq&\sum_{i=0}^{t-1}\frac{\delta}{q^{2i}}+\sum_{i=t}^{r-1}1\\
&=&\delta\left[\frac{1-q^{-2t}}{1-q^{-2}}\right]+(r-t)\\
&=&\left(\frac{q^{2t}-1}{q^{2t}-q^{2t-2}}\right)\delta+(r-t),
\end{IEEEeqnarray*}
so that
\[
n\geq\left(\frac{n+\kappa}{2r}-1\right)M
+\left(\frac{q^{2t}-1}{q^{2t}-q^{2t-2}}\right)\delta+(r-t).
\]
Dividing by $n$ and letting $n\to\infty$ yields
\begin{equation}\label{eq:aGriesmer}
\mathcal{R}\leq\frac{r-\rho+1}{M}
-\frac{2r}{M}\left(\frac{q^{2t}-1}{q^{2t}-q^{2t-2}}\right)\Delta+o(1).
\end{equation}
Similarly, substituting $\tau=\frac{n+\kappa}{2r}-1$ into the pure
\emph{Plotkin-like} bound \eqref{eq:qPlotkin} gives
\[
\delta\leq\left(\frac{q^{2r}-1}{q^{2r}-q^{2r-2}}\right)^{-1}
\left[n-\left(\frac{n+\kappa}{2r}-1\right)M\right],
\]
and hence, as $n\to\infty$,
\begin{equation}\label{eq:aPlotkin}
\mathcal{R}\leq\frac{r-\rho+1}{M}
-\frac{2r}{M}\left(\frac{q^{2r}-1}{q^{2r}-q^{2r-2}}\right)\Delta+o(1).
\end{equation}
Since
\begin{equation}\label{eq:monotone}
1\leq\frac{q^{2t}-1}{q^{2t}-q^{2t-2}}\leq\frac{q^{2r}-1}{q^{2r}-q^{2r-2}}
\qquad\text{for all } t\in[1,r],
\end{equation}
the pure \emph{Plotkin-like} bound \eqref{eq:aPlotkin} is tighter than the pure
\emph{Griesmer-like} bound \eqref{eq:aGriesmer}, which in turn is tighter than
the pure \emph{Singleton-like} bound \eqref{eq:aSingleton}. Comparing
\eqref{eq:aSingleton} with \eqref{eq:aGG} then gives the following hierarchy.

\begin{theorem}\label{th:hierarchy}
Let $\mathcal{Q}$ be a pure $[[n,\kappa,\delta]]_q$ qLRC with
$(r,\rho)$-locality obtained from the Hermitian construction of
Theorem~\ref{th:qLRCs}. If $r\geq3$, then the asymptotic pure
\emph{Plotkin-like} \eqref{eq:aPlotkin}, pure \emph{Griesmer-like}
\eqref{eq:aGriesmer}, pure \emph{Singleton-like} \eqref{eq:aSingleton} and
\emph{GG Singleton-like} \eqref{eq:aGG} bounds satisfy the strict hierarchy
\begin{IEEEeqnarray*}{c}
\text{Pure Plotkin-like}\;\succ\;\text{Pure Griesmer-like}\\
\succ\;\text{Pure Singleton-like}\;\succ\;\text{GG Singleton-like},
\end{IEEEeqnarray*}
where $\succ$ denotes ``is strictly tighter than''.
\end{theorem}

Finally, we establish the asymptotic form of the \emph{sphere-packing-like}
bound. Recall the $q^{2}$-ary entropy function
\[
H_{q^2}(x)=
\begin{cases}
0, & x=0,\\[1ex]
\begin{aligned}[t]
&x\log_{q^2}(q^{2}-1)-x\log_{q^2}x\\
&\quad-(1-x)\log_{q^2}(1-x),
\end{aligned} & 0<x\leq1-q^{-2},
\end{cases}
\]
and let
\begin{equation}\label{eq:vol}
V_{q^2}(N,a)=\sum_{i=0}^{a}\binom{N}{i}(q^{2}-1)^{i}
\end{equation}
denote the volume of a $q^{2}$-ary Hamming ball of radius $a$. By
\cite[Lemma~2.10.3]{huffman2010fundamentals},
\begin{equation}\label{eq:entropy}
\lim_{N\to\infty}\frac{1}{N}\log_{q^2}V_{q^2}\bigl(N,\lfloor xN\rfloor\bigr)
=H_{q^2}(x),\qquad 0\leq x\leq1-q^{-2}.
\end{equation}
From \eqref{eq:qSphere}, for every admissible integer $\tau$,
\[
\kappa\leq n-2\tau(\rho-1)-2\log_{q^2}
V_{q^2}\left(n-\tau M,\left\lfloor\frac{\delta-1}{2}\right\rfloor\right).
\]
If $\tau=O(1)$, then dividing by $n$ and applying \eqref{eq:entropy} gives
\begin{equation}\label{eq:aSPhigh}
\mathcal{R}\leq1-2H_{q^2}\left(\frac{\Delta}{2}\right)+o(1);
\end{equation}
that is, when $\tau$ remains bounded, the locality-dependent term
$2\tau(\rho-1)/n$ vanishes and the bound reduces to the standard asymptotic
quantum sphere-packing bound.

To capture the asymptotic effect of locality, consider a sequence of admissible
integers $\{\tau_n\}$ with $\lim_{n\to\infty}\tau_n/n=\lambda$, where
$0\leq\lambda<1/M$, and set $N_n=n-\tau_nM$. Then $N_n/n\to1-\lambda M$ and
$\frac{1}{N_n}\lfloor\frac{\delta-1}{2}\rfloor\to\frac{\Delta}{2(1-\lambda M)}$.
Hence, provided that
\begin{equation}\label{eq:entdomain}
0\leq\frac{\Delta}{2(1-\lambda M)}\leq1-q^{-2},
\end{equation}
the estimate \eqref{eq:entropy} gives
\[
\mathcal{R}\leq1-2\left[\lambda(\rho-1)+(1-\lambda M)
H_{q^2}\left(\frac{\Delta}{2(1-\lambda M)}\right)\right],
\]
and consequently
\begin{equation}\label{eq:aSPgeneral}
\mathcal{R}\leq1-2\max_{\lambda\in\Lambda_\Delta}
\left\{\lambda(\rho-1)+(1-\lambda M)
H_{q^2}\left(\frac{\Delta}{2(1-\lambda M)}\right)\right\},
\end{equation}
where $\Lambda_\Delta$ is the set of asymptotically admissible shortening ratios,
i.e.\ those $\lambda\in[0,1/M)$ that satisfy \eqref{eq:entdomain} and are
compatible with the range \eqref{eq:qtau}. Define
\[
g(\lambda)=\lambda(\rho-1)+(1-\lambda M)
H_{q^2}\left(\frac{\Delta}{2(1-\lambda M)}\right).
\]
For $\Delta>0$ and with the entropy argument in the interior of its domain, put
$N(\lambda)=1-\lambda M$ and $C=\Delta/2$. Using the identity
\[
\frac{d}{dN}\left[N H_{q^2}\left(\frac{C}{N}\right)\right]
=-\log_{q^2}\left(1-\frac{C}{N}\right),
\]
we obtain
\begin{equation}\label{eq:gprime}
g'(\lambda)=(\rho-1)+M\log_{q^2}
\left(1-\frac{\Delta}{2(1-\lambda M)}\right),
\end{equation}
and furthermore
\[
g''(\lambda)=-\frac{M^{2}\Delta}
{2\ln(q^{2})(1-\lambda M)^{2}
\left(1-\dfrac{\Delta}{2(1-\lambda M)}\right)}<0 .
\]
Thus $g$ is strictly concave for $\Delta>0$ in the interior of the entropy
domain, and any interior stationary point is its unique maximizer. Assume
henceforth that $\rho>1$ and define $\beta=q^{-2(\rho-1)/M}$. Setting
$g'(\lambda)=0$ in \eqref{eq:gprime} yields the stationary point
\[
\lambda^{*}=\frac{1}{M}\left(1-\frac{\Delta}{2(1-\beta)}\right),
\]
and $\lambda^{*}>0$ is equivalent to $\Delta<2(1-\beta)$. Accordingly, define
the critical relative distance
\begin{equation}\label{eq:deltacrit}
\Delta_{\mathrm{crit}}=2(1-\beta)
=2\left[1-q^{-\frac{2(\rho-1)}{r+\rho-1}}\right].
\end{equation}
If $\Delta\geq\Delta_{\mathrm{crit}}$, then $g'(0)\leq0$; since $g$ is concave,
its maximum is attained at $\lambda=0$ and \eqref{eq:aSPhigh} follows.

Now let $0<\Delta<\Delta_{\mathrm{crit}}$. At $\lambda=\lambda^{*}$ we have
$1-\lambda^{*}M=\frac{\Delta}{2(1-\beta)}$ and
$\frac{\Delta}{2(1-\lambda^{*}M)}=1-\beta$, so that
\[
g(\lambda^{*})=\frac{\rho-1}{M}\left(1-\frac{\Delta}{2(1-\beta)}\right)
+\frac{\Delta}{2(1-\beta)}H_{q^2}(1-\beta).
\]
Using $\log_{q^2}\beta=-\frac{\rho-1}{M}$ together with
$H_{q^2}(1-\beta)=(1-\beta)\log_{q^2}(q^{2}-1)-(1-\beta)\log_{q^2}(1-\beta)
-\beta\log_{q^2}\beta$, this simplifies to
\[
g(\lambda^{*})=\frac{\rho-1}{M}
+\frac{\Delta}{2}\log_{q^2}\left(\frac{\beta(q^{2}-1)}{1-\beta}\right).
\]
Substituting into \eqref{eq:aSPgeneral} gives
\[
\mathcal{R}\leq1-\frac{2(\rho-1)}{M}
-\Delta\log_{q^2}\left(\frac{\beta(q^{2}-1)}{1-\beta}\right).
\]
Since $1-\frac{2(\rho-1)}{M}=\frac{r-\rho+1}{r+\rho-1}$ and
$\frac{\beta}{1-\beta}=\bigl(q^{2(\rho-1)/(r+\rho-1)}-1\bigr)^{-1}$, we arrive
at the asymptotic \emph{sphere-packing-like} bound
\begin{equation}\label{eq:aSPfinal}
\mathcal{R}\leq\frac{r-\rho+1}{r+\rho-1}
-\Delta\log_{q^2}\left(\frac{q^{2}-1}
{q^{\frac{2(\rho-1)}{r+\rho-1}}-1}\right),
\quad 0<\Delta<\Delta_{\mathrm{crit}}.
\end{equation}
For $\Delta=0$ the function $g$ is linear, so the strict concavity argument does
not apply; however, \eqref{eq:aSPfinal} extends continuously to $\Delta=0$.
Recalling that the entropy function requires $\Delta/2\leq1-q^{-2}$, the two
regimes are summarized as
\begin{equation}\label{eq:aSPpiecewise}
\mathcal{R}\leq
\begin{cases}
\begin{aligned}[t]
&\frac{r-\rho+1}{r+\rho-1}-\Delta\log_{q^2}\Biggl(\\
&\quad\frac{q^{2}-1}{q^{\frac{2(\rho-1)}{r+\rho-1}}-1}\Biggr),
\end{aligned} & 0\leq\Delta<\Delta_{\mathrm{crit}},\\[2.5ex]
1-2H_{q^2}\left(\dfrac{\Delta}{2}\right)+o(1),
& \Delta_{\mathrm{crit}}\leq\Delta\leq2(1-q^{-2}),
\end{cases}
\end{equation}
with $\Delta_{\mathrm{crit}}$ as in \eqref{eq:deltacrit}.

The numerical comparison in Fig.~\ref{fig:asymptotic} confirms the ordering of
Theorem~\ref{th:hierarchy} and shows that, for certain parameter choices, the
bound \eqref{eq:aSPpiecewise} provides the tightest asymptotic constraint among
the bounds considered.

\begin{figure}[!t]
\centering
\includegraphics[width=3.3in]{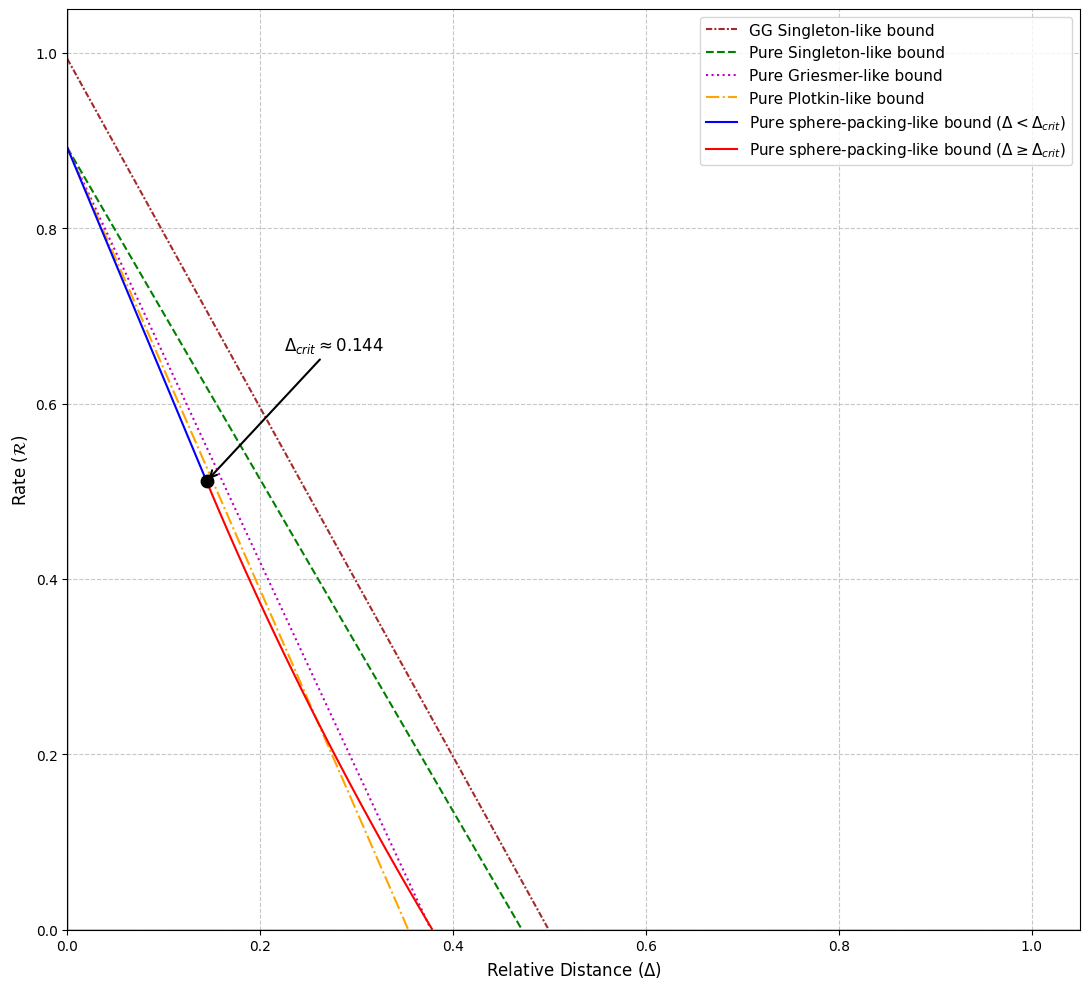}
\caption{Comparison of the asymptotic bounds and their relative tightness for an
$(r,\rho)=(280,17)$-qLRC over $\mathbb{F}_2$ with $t=2$.}
\label{fig:asymptotic}
\end{figure}

\section{Conclusion}\label{sec:conclusion}

We derived alphabet-dependent \emph{Griesmer-like},
\emph{sphere-packing-like} and \emph{Plotkin-like} upper bounds for pure
$(r,\rho)$-qLRCs obtained from Hermitian dual-containing classical
$(r,\rho)$-LRCs via the CSS construction. We showed that, asymptotically and for
$r\geq3$, the \emph{Plotkin-like} bound \eqref{eq:aPlotkin} is strictly tighter
than the \emph{Griesmer-like} bound \eqref{eq:aGriesmer}, which is in turn
strictly tighter than the existing \emph{Singleton-like}
\eqref{eq:aSingleton} and \emph{GG Singleton-like} \eqref{eq:aGG} bounds. A
numerical comparison of the \emph{sphere-packing-like} bound
\eqref{eq:aSPpiecewise} against this hierarchy over the
$(\mathcal{R},\Delta)$ plane is illustrated in Fig.~\ref{fig:asymptotic}.
Constructions attaining the bounds \eqref{eq:qGriesmer}, \eqref{eq:qSphere} and
\eqref{eq:qPlotkin} are left to the extended version of this work.

\bibliographystyle{IEEEtran}
\bibliography{references}

\end{document}